\documentclass{llncs}
\usepackage[utf8x]{inputenc}

\usepackage{listings}
\usepackage{bcprules}
\usepackage{color}
\usepackage{stmaryrd}
\usepackage{graphicx}
\usepackage{subfigure}
\usepackage{amsmath}
\usepackage{amssymb}
\usepackage{xypic}
\usepackage{cancel}
\usepackage{soul}
\usepackage[normalem]{ulem}

\title{Inferring Fences in a Concurrent Program Using SC proof of Correctness}
\author{Chinmay Narayan, Shibashis Guha, S.Arun-Kumar }
\institute{Indian Institute of Technology Delhi}
\newcommand{\program}{\mathsf{Program}}
\newcommand{\command}{\mathsf{Cmd}}
\newcommand{\tid}{\mathsf{Tid}}
\newcommand{\dom}{\mathsf{dom}}
\newcommand{\finish}{\mathbf{end}}
\newcommand{\shvar}{\mathsf{shvar}}
\newcommand{\locvar}{\mathsf{\ell var}}
\newcommand{\locexp}{\mathsf{\ell exp}}
\newcommand{\shexp}{\mathsf{shexp}}
\newcommand{\val}{\mathsf{val}}
\newcommand{\Assign}[2]{\mathrm{#1}\mathbf{:=}\mathrm{#2}}
\newcommand{\var}[1]{\mathrm{#1}}
\newcommand{\ifthenelse}[3]{\mathbf{if}\ {#1}\ \mathbf{then}\ {#2}\ \mathbf{else}\ {#3}}
\newcommand{\while}[2]{\mathbf{while}\ {#1}\ \mathbf{do}\ {#2}\ \mathbf{od}}                
\renewcommand{\skip}{\mathbf{skip}}
\newcommand{\sync}{\mathbf{fence}} 

\newcommand{\join}[1]{\mathbf{join}(#1)}
\renewcommand{\exp}{\mathsf{exp}}
\newcommand{\defeq}{\stackrel{def}{=}}

\newcommand{\scconfig}[2]{({#1},{#2})}
\newcommand{\G}{\mathcal{G}}
\newcommand{\T}[3]{\set{(#1,#2,#3)}}
\renewcommand{\L}{\mathcal{L}} 
\newcommand{\TStore}{\mathsf{Tstore}}
\newcommand{\sctrans}{\rightarrow}   
\newcommand{\denot}[2]{\llbracket #1 \rrbracket_{#2}}

\newcommand{\rmconfig}[2]{({#1},{#2})}
\renewcommand{\L}{\mathcal{L}} 
\newcommand{\B}{\mathcal{B}}
\newcommand{\rmtrans}{\stackrel{M}{\rightarrow}}   
\newcommand{\denotrm}[2]{\llbracket #1 \rrbracket^{rm}_{#2}}

\newcommand{\subst}[3]{#1[{#2}/{#3}]}
\newcommand{\Last}[1]{\mathsf{Last(#1)}}

\newcommand{\w}[2]{(!{\mathrm{#1}},\mathrm {#2})}
\renewcommand{\r}[2]{(?{\mathrm{#1}},\mathrm{#2})}

\newcommand{\true}{\mathsf{true}}
\newcommand{\false}{\mathsf{false}}
\newcommand{\set}[1]{\{#1\}}
\newcommand{\restr}[2]{#1\downharpoonleft_{#2}}

\newcommand{\len}[1]{|{#1}|}
\renewcommand{\Last}{\mathbf{Last}}
\newcommand{\None}{\mathbf{None}}

\newcommand{\Assn}{\mathsf{Assn}}

\newcommand{\HistAssn}{\mathsf{HAssn}}
\newcommand{\HA}{\mathsf{HA}}
\newcommand{\LStateAssn}{\mathsf{LSAssn}}
\newcommand{\LA}{\mathsf{LA}}
\newcommand{\ph}{\mathbf{ph}}

\newcommand{\State}[2]{\mathsf{S}(#1,#2)}
\newcommand{\Assndenot}[1]{\llbracket#1\rrbracket}
\newcommand{\Histdenot}[1]{{\textcolor{red}{\llbracket}#1\textcolor{red}{\rrbracket}}}
\newcommand{\Lstdenot}[2]{{\textcolor{green}{\llbracket}#1\textcolor{green}{\rrbracket}_{#2}}}

\newcommand{\bind}{\mathsf{bind}}
\newcommand{\bindsubst}[2]{_{#2}{#1}}
\newcommand{\any}{\_}

\newcommand{\Let}{\mathbf{Let~}}
\newcommand{\In}{\mathbf{~in}}

\newcommand{\inpar}[1]{\left(\begin{array}{@{}l@{}}#1\end{array}\right)}
\newcommand{\inarr}[1]{\begin{array}{@{}l@{}}#1\end{array}}
\newcommand{\CWhile}[1]{\textcolor{blue}{\mathbf{while}}(#1)~\textcolor{blue}{\mathbf{do}}}
\newcommand{\COd}{\textcolor{blue}{\mathbf{od}}}

\newcommand{\assh}[1]{\mathrm{\textcolor{magenta}{#1}}}
\newcommand{\CS}{\mathrm{\textcolor{red}{\fbox{Critical~Section}}}}

\newcommand{\sou}{\cite{Soundararajan198413}~}
\newcommand{\HOARE}[3]{\{{#1}\}\ #2\ \{{#3}\}}
\newcommand{\h}[1]{\mathsf{h_{#1}}}
\newcommand{\Proc}[1]{\mathsf{Proc}_{#1}}

\newcommand{\limp}{\Rightarrow}
\newcommand{\fence}{\mathbf{fence}}

\newcommand{\compat}{\mathsf{Compat}}
\newcommand{\before}{\prec}

\newcommand{\lock}{\mathrm{lock}}
\newcommand{\unlock}{\mathrm{unlock}}
\newcommand{\lab}[1]{{\scriptsize{#1.}}~~}

\newcommand{\Data}[1]{\mathrm{D^{#1}}}
\newcommand{\Seq}[1]{\mathrm{S_{#1}}}
\begin{document}
\maketitle
\begin{abstract}
Most proof systems for concurrent programs \cite{Jones83a} \cite{csl} \cite{rgsep} assume the underlying 
memory model to be sequentially consistent(SC), an assumption which does not hold for modern multicore processors. 
These processors, for performance reasons, implement relaxed memory models. As a result
of this relaxation a program, proved correct on the SC memory model, might execute incorrectly.
To ensure its correctness under relaxation, $\fence$ instructions are inserted in the code.
In this paper we show that the SC proof of correctness of an algorithm, carried out in the proof system of \cite{Soundararajan198413}, 
identifies per-thread instruction orderings \emph{sufficient} for this SC proof. 
Further, to correctly execute this algorithm on an underlying relaxed memory model it is sufficient to 
respect only these orderings by inserting $\fence$ instructions.
\end{abstract}

\section{Introduction}

The memory model of a processor defines the order in which memory operations, issued by a single processor, appear to execute from the point of view of the memory subsystem. 
In a broad sense, it determines whether any two memory access instructions issued by a processor within a single thread can be reordered. 
Sequentially Consistent memory model (SC) is the simplest but most restrictive of all and does not allow any reordering of instructions within a thread.
Modern multicore processors, for the purpose of hiding latencies, implement relaxed memory models and allow instructions within a thread to be reordered as long 
as they operate on different memory addresses.
For example, Total Store Order (TSO), the memory model for x86 processors, allows write instructions to get reordered with later 
reads provided they operate on different memory locations. 
As a result of these relaxations, a program may exhibit more behaviours than under SC
and it is possible that some of these extra behaviours do not satisfy the property which holds under SC.
Peterson's mutual exclusion algorithm, in Figure \ref{peterson} illustrates this behaviour.
This algorithm satisfies mutual exclusion property under the SC model but executing it on an Intel's x86 processor might result in the violation of this property.
This can happen if the read of flag$_2$ at label $3$ is reordered before instructions at label $1$ and label $2$. With such reordering $\Proc{1}$ can enter the critical section.
$\Proc{2}$, with or without this reordering, can also enter the critical section simultaneously. This example clearly shows the effect of memory model on the correctness of an algorithm.

\newcommand{\flag}{\var{flag}}
\newcommand{\turn}{\var{turn}}
\begin{figure}[h]
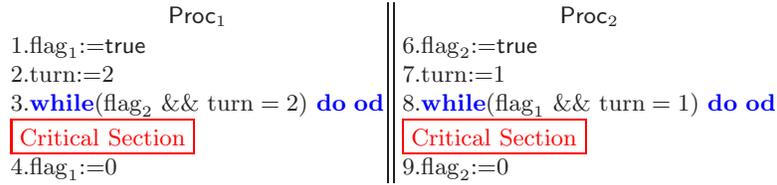

\centering
$\begin{array}{c||@{~}c}
\Proc{1} & \Proc{2}\\
\inarr{
1. \Assign{\flag_1}{\true}\\
2. \Assign{\turn}{2}\\
3. \CWhile{\flag_2~ \&\&~ \turn=2} 
   ~\COd\\
 \CS\\
4. \Assign{\flag_1}{0}\\
}&
\inarr{
6. \Assign{\flag_2}{\true}\\
7. \Assign{\turn}{1}\\
8. \CWhile{\flag_1~ \&\&~ \turn=1}
   ~\COd\\
 \CS\\
9. \Assign{\flag_2}{0}\\
}  
\end{array}$
\caption{Peterson's mutual exclusion algorithm}
\label{peterson}
\end{figure}
It is clear that this problem appeared because of an extra execution
generated due to instructions' reordering which was not possible under the SC memory model. 
This problem also does not appear for a data race free program \emph{if} the underlying relaxed memory model satisfies data-race freedom (DRF) property.
It can be shown that any data race free program when executed on a memory model satisfying DRF 
exhibits exactly the same set of behaviours as under the SC memory model. 
A program is \textit{data race free} if in every execution of this
program any pair of an conflicting instructions ($w\mbox{-}w$ or $w\mbox{-}r$) by two different threads to the same variable are separated by an $\unlock$ instruction. 
Therefore a \textit{data race free} and a \textit{correct} program under SC is guaranteed 
to execute correctly on a memory model which satisfies DRF property. 
However, the algorithms that we are interested in (lock-free, wait-free, lock implementations) do not use $\lock/\unlock$ and hence do not fit under this definition of \textit{data race freedom}. Therefore, their correctness under a relaxed memory model does not follow from their correctness under the SC memory model.

Another way to avoid extra executions is to prevent certain reorderings by putting a special instruction, $\fence$, after every instruction in each thread. 
A $\fence$ instruction when placed between any two instructions in a thread
prohibits their reordering. Execution of a fully fenced program ($\fence$ after every instruction in a thread)
on a relaxed memory model generates exactly the same set of executions as under SC and therefore the correctness under SC implies the correctness under relaxed memory model.
However, this trivial placement strategy would negate the performance benefits associated with relaxed memory models. 
Therefore, an ideal placement of $\fence$ instructions should preserve 
only those program orders which are \textit{sufficient} to prove the correctness of the properties of interest.

In this paper, we deal with parallel programs which satisfy some property under SC 
but are not race-free. 
The main contribution of this paper is to show that the proof of correctness of these programs under SC is useful in identifying
per-thread instruction orderings sufficient to make this program correct on a relaxed memory model. Further, locations of $\fence$ 
instructions can be inferred based on these orderings and the underlying memory model.
We are not aware of any other attempt to use the SC proof of correctness for \textit{fence inference} in a concurrent program.

 \section{Related Work}
\label{sec:related}

All existing approaches for inferring fences for relaxed memory models can be divided into two main categories; model checking based approaches 
\cite{Huynh:2007:MMS:1315662.1315668} 
\cite{Kuperstein:2011:PAR:1993498.1993521} 
\cite{Linden:2011:VAM:2032692.2032707} 
\cite{DBLP:dblp_conf/tacas/AbdullaACLR12} and proof system based approaches \cite{Ridge} \cite{bornat2012barrier} \cite{bornat2012abducing}. Model checking based approaches first explore the state space of 
a program under a given memory model using buffer based operational semantics and check the reachability of erroneous states.
Once a reachable erroneous state is identified, the path leading to this state is restricted by inserting fences at appropriate places. \cite{DBLP:conf/esop/AtigBBM12} and \cite{Atig:2010:VPW:1707801.1706303} showed that the state reachability problem for TSO and PSO memory models is decidable for finite state programs. 
Further, this problem becomes undecidable as soon as the \textit{read after write} reordering is added to the memory model. This approach, by its nature, is better suited to programs
with finite data domains. We are aware of only one line of work 
\cite{DBLP:conf/sas/AbdullaACLR12} which combines predicate abstraction and model checking based approach to verify and correct infinite data 
domain programs like Lamport's bakery algorithm. 

The second approach is to use a memory model specific proof system as done in \cite{Ridge}. \cite{Ridge} presents a separation logic  based proof system for the TSO memory model and shows that Simpson's 4 slot algorithm does not satisfy the interference freedom property. Recently \cite{bornat2012barrier} and \cite{bornat2012abducing} looked at the use of separation logic derived proof system for verifying concurrent data structures on POWER/ARM based memory models. These memory models are more complex than the TSO or the PSO memory model mainly because of non-atomic writes. 
Unlike the approaches of \cite{Ridge}, \cite{bornat2012barrier} and \cite{bornat2012abducing}, we do not propose a memory model specific proof system but only look at the proof of correctness under SC memory model and use it to infer \textit{sufficient} orderings required for the correctness. Unlike model checking approaches, proof system based approaches can cover more than just reachability. This is evident in the example of Simpson's 4 slot algorithm where apart from the interference freedom we also prove that the sequence of values observed by the reader are consistent, i.e. they form a stuttering sequence of the values written by the writer. We are not aware of any line of work which handled the fence inference in Simpson's 4 slot algorithm under PSO memory model with respect to the interference freedom and the consistent reads properties. 


\section{Language: Syntax and Semantics}
\label{sec:lang}
\newcommand{\langsyntax}{
\begin{figure}[h]
 \begin{align*}
  \program &= \command;\finish\\
  \command &= \Assign{\shvar}{\locexp} \mid \Assign{\locvar}{\locexp} \mid \Assign{\locvar}{\shexp} \mid \command_1;\command_2 \mid \program_1 \| \cdots \| \program_n \\
           & \mid \ifthenelse{\locexp}{\command_1}{\command_2} \mid \while{\locexp}{\command} \mid \skip \mid \sync\\
  \shexp &= \shvar \mid \shvar \oplus \locexp \\
  \locexp &= \locvar \mid \locexp \oplus \locexp \mid \val \\
  \val    &= \mathbb{N} \mid \mathbb{B} \mid \mathbb{T} \\
  \oplus &= + \mid - \mid \times \mid \div
 \end{align*}
\caption{A simple concurrent programming language with barrier support}
\label{fig:langsyntax}
\end{figure}
}

\newcommand{\refsemantics}{
\begin{figure*}[h]
$\begin{array}{l@{~}|@{~}r}
\begin{minipage}{0.57\textwidth}
\small
\infrule[GW]{
C = \Assign{\shvar}{\locexp};C' \\ \denot{\locexp}{\G,\L}=v \quad \G'=\G[\shvar:=v]
}{
\scconfig{\G}{\T{t}{\L}{C}\cup T} \sctrans \scconfig{\G'}{\T{t}{\L}{C'}\cup T}\\
}
\vspace{.2em}

\infrule[LRW]{
C=\Assign{\locvar}{\locexp};C' \\ \denot{\locexp}{\G,\L}=v \quad \L'=\L[\locvar:=v]
}{
\scconfig{\G}{\T{t}{\L}{C}\cup T} \sctrans \scconfig{\G}{\T{t}{\L'}{C'}\cup T}
}

\vspace{.2em}

\infrule[Join]{
C=\join{\locvar};C' \\ \denot{\locvar}{\G,\L}=t'\in \tid \quad t'\not\in \dom(T) \cup \set{t}
}{
\scconfig{\G}{\T{t}{\L}{C}\cup T} \sctrans \scconfig{\G}{\T{t}{\L}{C'}\cup T}
}

\vspace{.2em}

\infrule[ITE-T]{
\denot{\locexp}{\G,\L}=\true \\ C=\ifthenelse{\locexp}{\command_1}{\command_2};C' 
}
{
\scconfig{\G}{\T{t}{\L}{C}\cup T} \sctrans \scconfig{\G}{\T{t}{\L}{\command_1;C'}\cup T} 
}

\vspace{.2em}

\infrule[WHL-T]{
\denot{\locexp}{\G,\L}=\true \\ C=\while{\locexp}{\command};C'
}
{
\scconfig{\G}{\T{t}{\L}{C}\cup T} \sctrans \scconfig{\G}{\T{t}{\L}{\command;C}}
}
\end{minipage}&
\begin{minipage}{0.52\textwidth}
\small
\infrule[GR]{
C=\Assign{\locvar}{\shexp};C' \\ \denot{\shexp}{\G,\L}=v \quad \L'=\L[\locvar:=v]
}{
\scconfig{\G}{\T{t}{\L}{C}\cup T} \sctrans \scconfig{\G}{\T{t}{\L'}{C'}\cup T}
}
\vspace{1em}

\infrule[END]{
C=\finish
}{
\scconfig{\G}{\T{t}{\L}{C}\cup T} \sctrans \scconfig{\G}{T}
}
\vspace{1em}

\infrule[SKP-SYC]{
C\in\set{\skip;C',\sync;C'}
}{
\scconfig{\G}{\T{t}{\L}{C}\cup T} \sctrans \scconfig{\G}{\T{t}{\L}{C'} \cup T}
}

\vspace{.2em}
\infrule[ITE-F]{
\denot{\locexp}{\G,\L}=\false \\ C=\ifthenelse{\locexp}{\command_1}{\command_2};C' 
}
{
\scconfig{\G}{\T{t}{\L}{C}\cup T} \sctrans \scconfig{\G}{\T{t}{\L}{\command_2;C'}\cup T} 
}

\vspace{.2em}

\infrule[WHL-F]{
\denot{\locexp}{\G,\L}=\false \\ C=\while{\locexp}{\command};C'
}
{
\scconfig{\G}{\T{t}{\L}{C}\cup T} \sctrans \scconfig{\G}{\T{t}{\L}{C'}}
}
\end{minipage}
\end{array}$
\infrule[PARCOMP]{
T'=T\cup \set{\T{t_1}{\emptyset}{\program_1},\cdots,\T{t_n}{\emptyset}{\program_n}} \\
t_1,\cdots,t_n \not\in \dom(T)\cup\set{t}\\
C=\program_1\|\cdots\|\program_n ;C' 
}{
\scconfig{\G}{\T{t}{\L}{C}\cup T} \sctrans \scconfig{\G}{\T{t}{\L}{\join{t_1};\cdots;\join{t_n};C'}\cup T'}
}

\caption{Reference semantics (SC) for programming language of Figure \ref{fig:langsyntax}}
\label{fig:refsem}
\end{figure*}
}

\langsyntax
Figure \ref{fig:langsyntax} shows the syntax of a simple parallel programming language without the support of 
dynamic thread creation. Operator $\|$ is used to compose a finite number of programs in parallel. We explicitly distinguish 
local variables, ranged over by $\locvar$ and accessed only within a thread, and shared variables, ranged over by $\shvar$ and 
accessed by more than one thread. A local expression, ranged over by $\locexp$, is constructed using only local variables, values 
and operators. A shared expression, ranged over by $\shexp$, is constructed from \textit{exactly one} shared variable, 
another local expression and operators. Assignment command only allows assigning a local expression to a local variable, 
a shared expression to a local variable or a local expression to a shared variable.
Assignment of a shared expression to a shared variable can be broken down into assignments of one of the above forms. Because of 
this restriction every assignment command either reads at most one shared variable or writes to at most one 
shared variable but not both. This guarantees at most one memory load or store event per assignment expression which is 
helpful in reasoning about memory model and associated events.

\refsemantics

\subsection{SC Semantics}
Figure \ref{fig:refsem} shows the semantics of this language under SC. A state or configuration under SC 
is of the form $\scconfig{\G}{\TStore}$ where
\[
\begin{array}{lcl}
 \G \defeq \shvar \to \val
,&~~
\TStore \defeq \tid \to \L \times \program
,&~~
\L \defeq \locvar \to \val 
\end{array}
\]
Local store $\L$ and global store $\G$ maps local and shared variables respectively to their values.
Each thread is represented by a unique thread id.
Thread store maps a thread id to its local store and the program to be executed next.
In our operational semantics we use the set representation of this function, i.e $\TStore$ as a set of tuples of the form $(t,\L,C)$.
Function $\denot{-}{\G,\L} \in \exp \to \G \to \L \to \val$ such that
\[
\begin{array}{lcl}
 \denot{\shvar}{\G,\L} = \G(\shvar)
,&~~
\denot{\locvar}{\G,\L} = \L(\locvar)
,&~~
\denot{e_1\oplus e_2}{\G,\L} = \denot{e_1}{\G,\L} \oplus \denot{e_1}{\G,\L} 
\end{array}
\]
takes an expression 
and evaluates it to a value based on the mapping of variables in global store $\G$
and thread local store $\L$. This function is then used to define reference semantics in Figure \ref{fig:refsem}.
For any function $\mathcal{F}:A\to B$, $a'\in A,b'\in B$, we write $\mathcal{F}[a':=b']$ to denote a function which is same as $\mathcal{F}$ everywhere except at $a'$ where it evaluates to $b'$.
Semantic rules corresponding to the conditional and the looping constructs ($\mathrm{ITE\mbox{-}T,ITE\mbox{-}F,WHL\mbox{-}T,WHL\mbox{-}F}$), local variable's read write ($\mathrm{LRW}$) and global variable's read write ($\mathrm{GR,GW}$) are quite straightforward. In the parallel composition command, the parent thread stops its execution and waits for all children threads to finish their execution. This is achieved by adding a $\join$ command in the parent thread for each spawned thread. The rule for $\join{tid}$ command ensures that this command is executed when the thread corresponding to $tid$ has finished its execution. This also ensures that the parent thread waits for the completion of children threads before continuing further. $\sync$ and $\skip$ are like no-op under SC semantics. Following the syntax of Figure \ref{fig:langsyntax} one process $\Proc{i}$, or thread, can only execute one program $\program_i$. Therefore we sometime use $\Proc{i}$ and $\program_i$ interchangeably to mean the same thing. 
\section{Logic}
\label{sec:logic}

\newcommand{\compnonrec}
{
\begin{figure}[h]
$\begin{array}{ll}
\small
 \mathsf{Compat}(v_1,&\cdots,v_m,\h{\Proc{1}},\cdots,\h{\Proc{n}}) \defeq \\
&\begin{cases}
\shvar_1=v_1\ \wedge\ \cdots\ \shvar_m=v_m & \mbox{if } \h{\Proc{i}}=\epsilon,~i=1\cdots n
\\\exists h \in Merge(\h{\Proc{1}},\cdots,\h{\Proc{n}}).~& \mbox{otherwise} 
\\~~~ \forall j \le m. \shvar_j = f_j(v_j,h)\
\\~~~~~\land \forall k \le \len{h}.~ h[k]=\r{\shvar_j}{\ph_j} \limp  
\\~~~~~~\ph_j = f_j(v_j,h[1:k-1]), 
\end{cases}
\end{array}$\\
\caption{Non-recursive definition of $\mathsf{Compat}$ predicate}
\label{fig:compnonrec}
\end{figure}
}
\newcommand{\mergef}
{
\begin{figure}[h]
$\begin{array}{ll}
\small
\mathsf{Merge}&(\h{\Proc{1}},\cdots,\h{\Proc{n}}) \defeq\\  
&\begin{cases}
\{\epsilon\} & \mbox{if } \h{\Proc{i}}=\epsilon, i=1\cdots n\\
\{h\ |\ \exists i\le n.~ h_i \neq \epsilon\ \land h[1] = \h{\Proc{i}}[1]\\
~~~~\land \exists h' \in Merge(\h{\Proc{1}},\cdots,\h{\Proc{i-1}},rest(\h{\Proc{i}}),\cdots,\h{\Proc{n}})\\
~~~~\land h'=rest(h)\} & \mbox{otherwise}
\end{cases}
\end{array}$
\caption{Function $Merge$ generates the set of all possible interleaving of $\h{\Proc{1}}$ to $\h{\Proc{n}}$}
\end{figure}
}
\newcommand{\funf}
{
\begin{figure}[h]
\small
$\begin{array}{l}
 f_j(v,h) \defeq 
      \begin{cases} 
         v & \mbox{if } h=\epsilon \\ 
         f_j(v',rest(h))  & \mbox{if } first(h)=\w{\shvar_j}{v'} \\
         f_j(v,rest(h)) & \mbox{otherwise} 
      \end{cases} 
 \end{array}$
 \caption{Function to get the last value written to variable $\shvar_j$ in history $h$}
\end{figure}
}

\newcommand{\comprec}{
\begin{figure}[h]
\small
$\begin{array}{ll}
 \mathsf{Compat}(v_1,&\cdots,v_m,\h{\Proc{1}},\cdots,\h{\Proc{n}}) \defeq \\
&\begin{cases}
\shvar_1=v_1\ \wedge\ \cdots\ \shvar_m=v_m & \mbox{if } \h{\Proc{i}}=\epsilon,~i=1\cdots n
\\ \exists i,j.~ \fst(\h{\Proc{i}})=\r(\shvar_j,\ph_j) \land \mathsf{Compat}(v_1,\cdots,v_m,\h{\Proc{1}},\cdots,\rst{\h{\Proc{i}}},\cdots,\h{\Proc{n}}).~& \mbox{otherwise} 
\\~~\lor \exists i,j,v_j'~ \fst(\h{\Proc{i}})=\w(\shvar_j,v_j') \land \mathsf{Compat}(v_1,\cdots,v_j',\cdots,v_m,\h{\Proc{1}},\cdots,\rst{\h{\Proc{i}}},\cdots,\h{\Proc{n}})
\end{cases}
\end{array}$\\
\caption{Non-recursive definition of $\mathsf{Compat}$ predicate}
\end{figure}
}

\newcommand{\syntassertion}{
\begin{figure}[h]
\begin{align*}
 \Assn &= \HistAssn \land \LStateAssn\\
 \HistAssn &= \forall \phv.~ ?\shvar\rightarrow \phv.~ P(\phv) \mid \forall \phv.~ !\shvar\rightarrow \phv.~ P(\phv)\\
 & \mid \forall \phv.~ \Last(?\shvar)\rightarrow \phv.~ P(\phv) \mid \forall \phv.~ \Last(!\shvar)\rightarrow \phv.~ P(\phv)\\ 
 & \mid \HistAssn_1 \land \HistAssn_2 \mid \HistAssn_1 \lor \HistAssn_2 \mid [\HistAssn_1;\HistAssn_2] \\
 & \mid \HistAssn^* \mid \None(!\shvar) \mid \None(?\shvar) \mid \HistAssn_1 \implies \HistAssn_2 \\
 \LStateAssn &= Bexp \mid \LStateAssn_1 \land \LStateAssn_2 \mid \LStateAssn_1 \lor \LStateAssn_2 \\
 & \mid \LStateAssn_1 \implies \LStateAssn_2\\
 Bexp & = Aexp \oplus Aexp\\
 Aexp & = \locvar \mid \val \mid \phv \mid Aexp + Aexp \mid Aexp - Aexp\\
 \oplus &= > \mid < \mid \ge \mid \le \mid \neq \mid = \mid + \mid - \mid \\
\end{align*}
\end{figure}
}



\newcommand{\semassertion}{
\begin{figure}[h]
 \begin{align*}
\Assndenot{\HA \land \LA} =& \set{\State{ls}{\sigma} \mid \exists \bind.~ (\sigma,\bind)\in \Histdenot{\HA} \land \Lstdenot{\bindsubst{\LA}{\bind}}{ls}=\true} \\  
\Histdenot{\forall \phv.~ ?\shvar\rightarrow \phv.~ P(\phv)} =&\set{(\sigma,\phv/v)\mid \len{\sigma}=1 \land \sigma(0)=(?\shvar,v)} \\
\Histdenot{\forall \phv.~ !\shvar\rightarrow \phv.~ P(\phv)} =&\set{(\sigma,\phv/v)\mid \len{\sigma}=1 \land \sigma(0)=(!\shvar,v)} \\
\Histdenot{\forall \phv.~ \Last(?\shvar)\rightarrow \phv.~ P(\phv)} =&\set{(\sigma,\phv/v)\mid Last(\restr{\sigma}{?\shvar})=(?\shvar,v)} \\
\Histdenot{\forall \phv.~ \Last(!\shvar)\rightarrow \phv.~ P(\phv)} =&\set{(\sigma,\phv/v)\mid Last(\restr{\sigma}{!\shvar})=(?\shvar,v)} \\
\Histdenot{\None(!et)} =&\set{(\sigma,\any)\mid \restr{\sigma}{!\shvar}=\epsilon} \\
\Histdenot{\None(!ot)} =&\set{(\sigma,\any)\mid \restr{\sigma}{?\shvar}=\epsilon} \\
\Histdenot{[\HA_1;\HA_2]} =&\left(\begin{array}{ll}
(\sigma,\bind) \mid 
&  \exists \sigma_0,\sigma_1,\bind_1,\bind_2.~ \sigma=\sigma_0.\sigma_1 \land (\sigma_0,\bind_1)\in \Histdenot{\HA_1}\\
&  \land (\sigma_1,\bind_2)\in \Histdenot{\bindsubst{\HA_2}{\bind_1}} \land \bind=\bind_1\cup\bind_2 
                                  \end{array}\right)\\
\Histdenot{\HA^*} =&\epsilon \cup \set{(\sigma,\bind)\mid (\sigma,\bind) \in \Histdenot{\HA^+}} \\
\Histdenot{\HA^+} =&\set{(\sigma,\bind)\mid (\sigma,\bind) \in \Histdenot{\HA;\HA^*}} \\
\Histdenot{\HA_1 \implies \HA_2} =&\set{(\sigma,\bind)\mid (\sigma,\bind) \in \Histdenot{\HA_1} \implies (\sigma,\bind)\in \Histdenot{\HA_2}} \\
\Histdenot{\HA_1 \land \HA_2} =&\set{(\sigma,\bind)\mid (\sigma,\bind) \in \Histdenot{\HA_1} \land (\sigma,\bind)\in \Histdenot{\HA_2}} \\
\Lstdenot{\locvar}{ls} =& ls(\locvar)\\
\Lstdenot{v}{ls} =&v\\
\Lstdenot{ae_1 + ae_2}{ls} =& \Lstdenot{ae_1}{ls} + \Lstdenot{ae_2}{ls}\\
\Lstdenot{ae_1 - ae_2}{ls} =& \Lstdenot{ae_1}{ls} - \Lstdenot{ae_2}{ls}\\
\Lstdenot{\LA_1 \land \LA_2}{ls} = & \Lstdenot{\LA_1}{ls} \land \Lstdenot{\LA_2}{ls}\\
\Lstdenot{\LA_1 \lor \LA_2}{ls} = & \Lstdenot{\LA_1}{ls} \lor \Lstdenot{\LA_2}{ls}\\
\Lstdenot{\LA_1 \oplus \LA_2}{ls} = & \Lstdenot{\LA_1}{ls} \oplus \Lstdenot{\LA_2}{ls}\\
\Lstdenot{\LA_1 \implies \LA_2}{ls} = &\Lstdenot{\LA_1}{ls} \implies \Lstdenot{\LA_2}{ls}\\
 \end{align*}
\end{figure}
}

\newcommand{\defpredicates}{
\begin{figure}[h]
\begin{minipage}{.5\textwidth}
\begin{align*}
  \None?\shvar \defeq& \lambda elem, h.~ \restr{h}{\r{\shvar}{\any}} = \epsilon\\  
  \None!\shvar \defeq& \lambda elem, h.~ \restr{h}{\w{\shvar}{\any}} = \epsilon\\
\end{align*} 
\end{minipage}
\begin{minipage}{.5\textwidth}
\begin{align*}
\None! \defeq& \lambda elem, h.~ \restr{h}{\w{\any}{\any}} = \epsilon\\
 ~[P;Q]\defeq&  \lambda h. \exists h_0,h_1.~ h=h_0.h_1 \land P(h_0) \land Q(h_1)\\
\end{align*} 
\end{minipage}
\begin{align*}
 ~[P]^* \defeq&  \lambda h.~ h=\epsilon \lor \exists h_0,h_1.~ h=h_0.h_1 \land P(h_0) \land  [P]^*(h_1)\\
 ~[P]^+ \defeq&  \lambda h.~ \exists h_0,h_1.~ h_0\neq \epsilon \land h=h_0.h_1 \land P(h_0) \land [P]^*(h_1)\\
\end{align*}
\caption{Predicates on individual history variable $\h{i}$}
\label{fig:predhist}
\end{figure}
}

%
In the proof system of \sou every process $\Proc{i}$ executing a program $\program_i$ 
has a history variable $\h{\Proc{i}}$ which 
captures the interaction of this process with shared variables in terms of 
values read from and written to them. 
$\h{\Proc{i}}$ is a sequence of elements of the form $\r{\shvar}{\ph}$ or $\w{\shvar}{v}$. $\r{\shvar}{\ph}$ is 
added to the sequence when $\Proc{i}$ reads a value from the shared variable $\shvar$. Similarly, $\w{\shvar}{v}$ 
is added to the sequence when $\Proc{i}$ writes a value $\mathrm{v}$ to the shared variable $\shvar$.  
Local reasoning of a program $\program_i$ generates a triple of the form 
$\HOARE{P_i}{\program_i}{Q_i}$ where $P_i$ and $Q_i$ define assertions 
on the \textit{local state} of the process as well as on the \textit{history} $\h{\Proc{i}}$.
Rest of this section describes the axioms of this proof system explaining the idea of local 
reasoning in terms of history variable and the parallel composition rule in terms of $\mathsf{Compat}$ predicate.

\paragraph{Axioms} In our programming language 
all program constructs, except $\mathrm{GR}$ and $\mathrm{GW}$, operate on local expression and therefore do not require reading from or writing 
to shared variables. 
Therefore the proof rules for these constructs are same as the Hoare's axioms in sequential 
setting.
In the following proof rules, the notation 
$\subst{P}{Q}{Q'}\subst{}{R}{R'}$ denote simultaneous substitution of $Q'$ for $Q$ and $R'$ for $R$ in $P$. Operator $``."$ 
concatenates an element to a sequence and $\epsilon$ is the empty sequence. Given a sequence $\sigma$, $\len{\sigma}$ is the length and $\sigma[i]$ 
is the $i^{th}$ element of this sequence.
The proof rule for the assignment to shared variable is as following,
\[
 \tag{GWrite}
 \HOARE{\subst{P}{\h{\Proc{i}}}{\h{\Proc{i}}.\w{\shvar}{\locexp}}}{\Assign{\shvar}{\locexp}}{P}
\]
As a result of this write the history is appended with the element $\w{\shvar}{\locexp}$. The proof rule for the reading of a shared variable, 
$\Assign{\locvar}{\shexp_j}$, is given as
\[
 \tag{GRead}
 \HOARE{\forall \ph.~\subst{P}{\h{\Proc{i}}}{\h{\Proc{i}}.\r{\shvar_j}{\ph}}\subst{}{\locvar}{\subst{\shexp_j}{\shvar_j}{\ph}} }{\Assign{\locvar}{\shexp_j}}{P}
\]

This rule requires the value of the shared variable $\shvar_j$ in order to evaluate the 
expression $\shexp_j$ but while reasoning locally we do not know the value beforehand. Therefore instead of the
actual value of $\shvar_j$ a \textit{placeholder variable} $\ph$ is assigned to $\shvar_j$ and this information is stored in $\h{\Proc{i}}$ by appending it 
with $\r{\shvar_j}{\ph}$. Further, $\shexp_j$ is evaluated accordingly 
before being assigned to $\locvar$ in the assertion. Here $\ph$ is universally quantified to all possible
values. 

\infrule[ParComp]{
\HOARE{P_i \land \h{\Proc{i}}=\epsilon}{\program_i}{Q_i},~ i=1\cdots n
}{
\HOARE{ P_1\land \cdots \land P_n \land \shvar_1=v_1 \land \cdots \shvar_m=v_m}{\\\program_1\|\cdots\|\program_n\\}{Q_1 \land \cdots \land Q_n \land \mathsf{Compat}(v_1,\cdots,v_m,\h{\Proc{1}},\cdots,\h{\Proc{n}})}
}
In parallel composition, each process is analyzed in isolation with the initial value of 
its history variable $\h{\Proc{i}}$ set to empty and some precondition $P_i$ on its local state. This gives post-condition $Q_i$ for each process $\Proc{i}$ which also 
contains the assertions on $\h{\Proc{i}}$. Precondition of the parallel composition rule is the conjunctions of individual processes' preconditions 
and the initial values of shared variables $\shvar_1$ to $\shvar_m$. Post-condition of this rule is the conjunctions of individual processes' post-conditions 
along with the predicate $\mathsf{Compat}$ where $\mathsf{Compat}$ is defined in Figure \ref{fig:compnonrec}.
\compnonrec
$Merge(\h{\Proc{1}},\cdots,\h{\Proc{n}})$ represent the set of all possible interleavings of histories $\h{\Proc{1}}$ to $\h{\Proc{n}}$ such that the sequence of 
elements within them is preserved in the merged history. Function $f_j(v,h)$ returns the last 
value written to the variable $\shvar_j$ in history $h$. It returns $v$ if no such write is found.
Essentially, the predicate $\mathsf{Compat}$ generates a set of 
equality predicates (one corresponding to each merged history). The first line in $\mathsf{Compat}$'s definition denotes that the final 
value of any shared variable $\shvar_j$ is the last value written to $\shvar_j$ in that merged history. 
Second line relates the placeholder value $\ph_j$, corresponding to a read of a shared variable $\shvar_j$, 
to the latest value written to $\shvar_j$ just before this element in the merged history. 
This is sufficient to characterize all compatible merged histories and therefore plays the central 
role in proofs of \S\ref{sec:examples} and of Appendix \ref{lamport}.

Individual process histories contains placeholder variables for every read. In Figure \ref{fig:predhist} we define, rather informally, 
a set of predicates over these histories in order to succinctly represent them in our proofs. Given a sequence $\sigma$, $\restr{\sigma}{type}$ 
denote the restricted subsequence of $\sigma$ consisting of only $type$ elements. Predicate $\None!$ holds if the history 
does not contain any write to any shared variable. $\None!\shvar$ and $\None?\shvar$ hold if the history does not contain any write to or read from 
the variable $\shvar$. $[P]^*$ and $[P]^+$ capture the regularity of history sequence by abstracting the placeholder variables. 
We also admit $P^+\limp P^*$ as a relaxation on the history sequence which is used in the consequence rule.
We admit that the use of these predicates, without giving proper semantics,
is not fully justified but we do it solely for the purpose of making our proofs manageable. We leave more formal treatment of these predicates for the future work. 

\defpredicates

%

\section{Relaxed Memory Model}
\label{sec:relmodels}
\newcommand{\buffigure}
{
\begin{figure*}[t]
\subfigure[]{
\includegraphics[scale=.7]{./bufftso.eps}
\label{fig:bufftso}
}\qquad \qquad
\subfigure[]{
\includegraphics[scale=.7]{./psobuff.eps}
\label{fig:buffpso}
}
\label{fig:buffExample}
\caption[Optional caption for list of figures]{\subref{fig:bufftso} shows the buffer of a thread $t$ in TSO semantics, \subref{fig:buffpso} shows the buffers of a thread $t$ in PSO semantics}
\end{figure*}
}

\newcommand{\relsemantics}{
\begin{figure*}[h]
$\begin{array}{l@{~}|@{~}r}
\begin{minipage}{0.56\textwidth}
\small
\infrule[GW]{
C = \Assign{\shvar}{\locexp};C' \quad \B(\shvar)=\sigma \\ \denotrm{\locexp}{\G,\L,\B}=v \quad \B'=\B[\shvar := \sigma.{v}]
}{
\rmconfig{\G}{\T{t}{\L,\B}{C} \cup T} \rmtrans \rmconfig{\G}{\T{t}{\L,\B'}{C'}\cup T}
}
\vspace{.2em}

\infrule[LRW]{
C=\Assign{\locvar}{\locexp};C' \\ \denotrm{\locexp}{\G,\L,\B}=v \quad \L'=\L[\locvar:=v]
}{
\rmconfig{\G}{\T{t}{\L,\B}{C}\cup T} \rmtrans \rmconfig{\G}{\T{t}{\L',\B}{C'}\cup T}
}

\vspace{.2em}

\infrule[Join]{
C=\join{\locvar};C' \\ \denotrm{\locvar}{\G,\L,\B}=t' \in \tid \quad t'\not\in \dom(T) \cup \set{t}
}{
\rmconfig{\G}{\T{t}{\L,\B}{C}\cup T} \rmtrans \rmconfig{\G}{\T{t}{\L,\B}{C'}\cup T}
}

\end{minipage}&
\begin{minipage}{0.52\textwidth}
\small
\infrule[GR]{
C=\Assign{\locvar}{\shexp};C' \\ \denotrm{\shexp}{\G,\L,\B}=v \quad \L'=\L[\locvar:=v]
}{
\rmconfig{\G}{\T{t}{\L,\B}{C}\cup T} \rmtrans \rmconfig{\G}{\T{t}{\L',\B}{C'}\cup T}
}
\vspace{.2em}

\infrule[END]{
C=\finish \quad \forall \shvar \in \dom(\B).~ \B(\shvar)=\epsilon
}{
\rmconfig{\G}{\T{t}{\L,\B}{C}\cup T} \rmtrans \rmconfig{\G}{T}
}

\vspace{.2em}

\infrule[Flush]{
\exists \shvar \in \dom(\B).~ \B(\shvar)=v.\sigma \\ \B'=\B[\shvar:=\sigma] \quad \G'=\G[\shvar:=v]
}{
\rmconfig{\G}{\T{t}{\L,\B}{C}\cup T} \rmtrans \rmconfig{\G'}{\T{t}{\L,\B'}{C}\cup T}
}
\vspace{.2em}

\infrule[Fence]{
C=\fence;C' \\ \forall \shvar \in \dom(\B).~ \B(\shvar)=\epsilon 
}{
\rmconfig{\G}{\T{t}{\L,\B}{C}\cup T} \rmtrans \rmconfig{\G}{\T{t}{\L,\B}{C'}\cup T} \\
}

\end{minipage}
\end{array}$
\begin{minipage}{\textwidth}
\vspace{1em}
\infrule[PARCOMP]{
T'=T\cup \set{\T{t_1}{\emptyset}{\program_1},\cdots,\T{t_n}{\emptyset}{\program_n}} \\
t_1,\cdots,t_n \not\in \dom(T)\cup\set{t}\\
C=\program_1\|\cdots\|\program_n ;C' \quad \forall \shvar \in \dom(\B).~ \B(\shvar)=\epsilon
}{
\rmconfig{\G}{\T{t}{\L,\B}{C}\cup T} \rmtrans \rmconfig{\G}{\T{t}{\L,\B}{\join{t_1};\cdots;\join{t_n};C'}\cup T'}
}
\end{minipage}
\caption{Relaxed semantics for programming language of Figure \ref{fig:langsyntax}}
\label{fig:relsem}
\end{figure*}
}

\newcommand{\denrm}{
\begin{figure}[h]
\begin{align*}
\denotrm{\shvar}{\G,\L,\B} \defeq&
\begin{cases}
  v & \text{ if } \B(\shvar)=\sigma.v \\
 \G(\shvar)        & \text{ otherwise } 
\end{cases}
\end{align*}
$\denotrm{\locvar}{\G,\L,\B} = \L(\locvar) 
\qquad \qquad \qquad \qquad  \denotrm{e_1\oplus e_2}{\G,\L,\B} = \denotrm{e_1}{\G,\L,\B} \oplus \denotrm{e_1}{\G,\L,\B}$
\caption{$\denotrm{-}{\G,\L,\B}$  for relaxed semantics}
\label{fig:denrm}
\end{figure}
}

 \relsemantics 
  \denrm
For the rest of this paper we consider the PSO memory model which allows reordering of a write instruction with future reads and future write instructions operating on different variables. To simulate the effect of PSO, every thread is equipped with one buffer per shared variable. These buffers store the values written on the corresponding variable by this thread in a queue(FIFO) discipline. Buffering the value of the write in a variable specific queue, simulates the effect of delaying the execution of a write instruction past future reads and future writes on different memory locations. If a read instruction for any  shared variable $\shvar$ is executed in a thread then the local buffer of $\shvar$ is checked first. If this buffer is non empty then the latest written value (from the tail) is returned. In case of empty buffer the value is read from the global state. 
This memory model also provides an explicit $\fence$ instruction in order to restrict reordering of any two instructions within a thread. Operationally this is achieved by flushing all the buffers of that thread. 

For relaxed semantics we need some modifications in the notion of State. 
A state or configuration under this memory model is of the form $\rmconfig{\G}{\TStore}$ where,
\[
\begin{array}{lr}
 \G \defeq \shvar \to \val
&~~
\TStore \defeq \tid \to \L \times \B \times \program\\
\B \defeq \shvar \to \sigma
&~~
 \L \defeq \locvar \to \val 
\end{array}
\]
The only change with respect to SC state is in the definition of thread store. Here, the range of this function 
also contains a buffer store $\B$ which is a function from shared variable to an ordered sequence of values, ranged over by $\sigma$.
Function $\denotrm{-}{\G,\L,\B} \in \exp \to \G \to \L \to \B \to \val$, defined in Figure \ref{fig:denrm},
takes an expression and evaluates it based on the values stored in the global store, thread local store and buffer store,
Further, relaxed semantics is defined in Figure \ref{fig:relsem}. For the constructs not shown in Figure \ref{fig:relsem} the semantics is the same as in the SC semantics of Figure \ref{fig:refsem} except for the change in the evaluation function from $\denot{-}{\G,\L}$ to $\denotrm{-}{\G,\L,\B}$.
In relaxed semantics, a thread $t$ executing a write to a shared variable $\shvar$ enqueues the value in the buffer of $\shvar$ in $t$. Any read of a shared variable $\shvar$ returns the latest value in the buffer of $\shvar$, if any. If this buffer is empty then the read returns the value from the global state (memory). $\mathrm{Flush}$ operation non-deterministically deques an element from any thread buffer and updates the global state accordingly. Further, in the parallel composition rule the requirement for the parent thread's buffer being empty ensures that instructions before and after the parallel composition are ordered. Same holds for the $\finish$ command as well.

\section{Examples}
\label{sec:examples}
We use our proof system to prove the correctness of Lamport's bakery algorithm \cite{Lamportbakery} for two processes and Simpsons's 4 slot algorithm \cite{Simpson4slot}. 
\begin{itemize}
 \item Lamport's algorithm has unbounded data domain for token variables which makes its verification challenging for model checking based approaches. We prove that this algorithm satisfies the mututal exclusion property, i.e. it is never possible for both processes to be inside their critical section simultaneously.
 \item Simpson's 4 slot algorithm implements a wait-free and lock-free atomic register for concurrent reader and writer. This algorithm uses disjoint slots to read from and write to in presence of interference. We prove that this algorithm is safe in the sense that concurrent reader and writer never use the same slot in presence of interference. 
 Further, we also prove that the values observed by successive reads are in the same order as written by the writer.
\end{itemize}
One important notation used in our proofs is as following;
Let $P$ and $Q$ be the assertions about individual elements in a history sequence then  $P\before Q$ denotes the fact that the element satisfying the assertion $P$ appears before the element satisfying the assertion $Q$ in the history sequence.

\newcommand{\token}{\var{token}}
\newcommand{\taking}{\var{taking}}
\newcommand{\prothreeoone}{
\assh{Inv1\defeq \lambda h.~}\\
\assh{\inpar{
\Let ph_1,ph_2,ph_3,ph_4,\\
~~~~P_1,Q_1,R_1,T_1,U_1,V_1,W_1. \In~\\
P_1 = \lambda h.~ \w{token_1}{0}(h) \In\\
Q_1 = \lambda h.~ \w{taking_1}{\true}(h) \In\\
R_1 = \lambda h.~ \r{token_2}{ph_1}(h) \In\\
T_1 = \lambda h.~ \w{token_1}{ph_1+1}(h) \In\\
U_1 = \lambda h.~ \w{taking_1}{\false}(h) \In\\
V_1 = \lambda h.~ \r{taking_2}{ph_2}(h) \In\\
W_1= \lambda h.~  \r{token_1}{ph_3}(h) \In\\
X_1= \lambda h.~  \r{token_2}{ph_4}(h) \In\\
~[[Q_1;R_1;T_1;U_1;\None!;P_1]^*;Q_1;R_1;T_1;\\
~~~~~U_1;\None!;V_1;\None!;W_1;X_1](h)
}}
}

\newcommand{\prothreetwo}{
\assh{Inv2\defeq \lambda h.~}\\
\assh{\inpar{
\Let ph_1',ph_2',ph_3',ph_4',\\
~~~~P_1,Q_2,R_2,T_2,U_2,V_2,W_2. \In~\\
P_2 = \lambda h.~ \w{token_2}{0}(h) \In\\
Q_2 = \lambda h.~ \w{taking_2}{\true}(h) \In\\
R_2 = \lambda h.~ \r{token_1}{ph_1'}(h) \In\\
T_2 = \lambda h.~ \w{token_2}{ph_1'+1}(h) \In\\
U_2 = \lambda h.~ \w{taking_2}{\false}(h) \In\\
V_2 = \lambda h.~ \r{taking_1}{ph_2'}(h) \In\\
W_2= \lambda h.~  \r{token_1}{ph_3'}(h) \In\\
X_2= \lambda h.~  \r{token_2}{ph_4'}(h) \In\\
~[[Q_2;R_2;T_2;U_2;\None!;P_2]^*;Q_2;R_2;T_2;\\
~~~~~U_2;\None!;V_2;\None!;W_2;X_2](h)
}}
}

\newcommand{\FigureExIII}{
\begin{figure}[h]
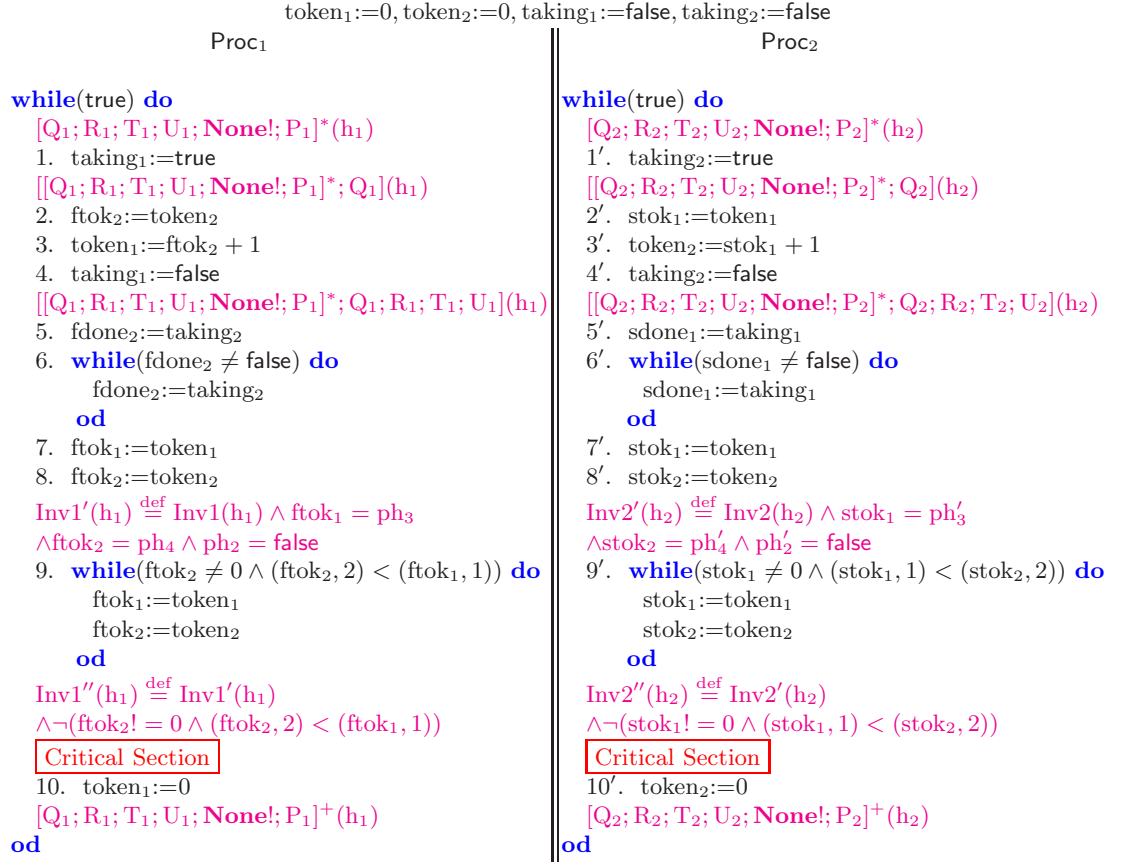

\centering
\small
$\begin{array}{c}
  \Assign{token_1}{0}, \Assign{token_2}{0}, \Assign{taking_1}{\false}, \Assign{taking_2}{\false}
\\
\begin{array}{l||l}
\inarr{
\begin{array}{r}
{\qquad \qquad \qquad \qquad {\Proc{1}}}\\ \\
\end{array}\\
\CWhile{\true}\\
\quad \assh{[Q_1;R_1;T_1;U_1;\None!;P_1]^*(h_1)}\\
\quad \lab{1} \Assign{taking_1}{\true}\\
\quad \assh{[[Q_1;R_1;T_1;U_1;\None!;P_1]^*;Q_1](h_1)}\\
\quad \lab{2} \Assign{ftok_2}{token_2}\\
\quad \lab{3} \Assign{token_1}{ftok_2+1}\\
\quad \lab{4} \Assign{taking_1}{\false}\\
\quad \assh{[[Q_1;R_1;T_1;U_1;\None!;P_1]^*;Q_1;R_1;T_1;U_1](h_1)}\\
\quad \lab{5} \Assign{fdone_2}{taking_2}\\
\quad \lab{6} \CWhile{\var{fdone_2}\neq \false}\\
\quad ~~~~ \quad \Assign{fdone_2}{taking_2}\\
\quad ~~~~~\COd\\
\quad \lab{7} \Assign{ftok_1}{token_1}\\
\quad \lab{8} \Assign{ftok_2}{token_2}\\
\quad \assh{Inv1'(h_1)\defeq Inv1(h_1) \land ftok_1 = ph_3 }\\
\quad \assh{\land ftok_2=ph_4 \land ph_2=\false}\\
\quad \lab{9}\CWhile{\var{ftok_2}\neq 0 \land (\var{ftok_2},2)<(\var{ftok_1},1)}\\
\quad~~~~\quad  \Assign{ftok_1}{token_1}\\
\quad~~~~\quad  \Assign{ftok_2}{token_2}\\
\quad ~~~~~\COd\\
\quad \assh{Inv1''(h_1)\defeq Inv1'(h_1)}\\
\quad \assh{\land \neg(ftok_2!=0 \land (ftok_2,2)<(ftok_1,1))}\\
\quad \CS \\
\quad \lab{10} \Assign{token_1}{0}\\
\quad \assh{[Q_1;R_1;T_1;U_1;\None!;P_1]^+(h_1)}\\
\COd
}
& 
\inarr{
\begin{array}{r}
{\qquad \qquad \qquad \qquad {\Proc{2}}}\\ \\
\end{array}\\
\CWhile{\true}\\
\quad \assh{[Q_2;R_2;T_2;U_2;\None!;P_2]^*(h_2)}\\
\quad \lab{1'}\Assign{taking_2}{\true}\\
\quad \assh{[[Q_2;R_2;T_2;U_2;\None!;P_2]^*;Q_2](h_2)}\\
\quad \lab{2'}\Assign{stok_1}{token_1}\\
\quad \lab{3'}\Assign{token_2}{stok_1+1}\\
\quad \lab{4'}\Assign{taking_2}{\false}\\
\quad \assh{[[Q_2;R_2;T_2;U_2;\None!;P_2]^*;Q_2;R_2;T_2;U_2](h_2)}\\
\quad \lab{5'}\Assign{sdone_1}{taking_1}\\
\quad \lab{6'}\CWhile{\var{sdone_1}\neq \false}\\
\quad ~~~~ \quad \Assign{sdone_1}{taking_1}\\
\quad ~~~~~ \COd\\
\quad \lab{7'}\Assign{stok_1}{token_1}\\
\quad \lab{8'}\Assign{stok_2}{token_2}\\
\quad \assh{Inv2'(h_2)\defeq Inv2(h_2) \land stok_1 = ph_3'}\\
\quad \assh{\land stok_2=ph_4' \land ph_2'=\false}\\
\quad \lab{9'}\CWhile{\var{stok_1}\neq0 \land (\var{stok_1},1)<(\var{stok_2},2)}\\
\quad~~~~\quad  \Assign{stok_1}{token_1}\\
\quad~~~~\quad  \Assign{stok_2}{token_2}\\
\quad ~~~~~\COd\\
\quad \assh{Inv2''(h_2)\defeq Inv2'(h_2)}\\
\quad \assh{ \land \neg(stok_1!=0 \land (stok_1,1)<(stok_2,2))}\\
\quad \CS\\
\quad \lab{10'}\Assign{token_2}{0}\\
\quad \assh{[Q_2;R_2;T_2;U_2;\None!;P_2]^+(h_2)}\\
\COd
}
\end{array}
\end{array}$
 \caption{Lamport's Bakery Algorithm for Two Processes}
 \label{fig:ex3}
\end{figure}
}

\subsection{Example: Lamport's Bakery Algorithm for Two Processes}
 
\begin{figure}[h]
 \begin{minipage}{0.5\textwidth}
  $\prothreeoone$
 \end{minipage}
\begin{minipage}{0.5\textwidth}
 $\prothreetwo$
\end{minipage}
\caption{Invariants $Inv1$ and $Inv2$ of Figure \ref{fig:ex3}}
\label{fig:ex3inv}
\end{figure}
\FigureExIII

In Lamport's algorithm each process $\Proc{i}$ operates on shared variables $token_i$ and $taking_i$ of type integer and boolean respectively. 
When a process $\Proc{i}$ intends to enter the critical section, it first reads the value of $token$ corresponding to another process, say $v$, 
and assigns the value $v+1$ to its own $token$ variable. $ftok_1$ and $ftok_2$ are local variables of $\Proc{1}$ which hold the value of $token_1$ and $token_2$ respectively. Similarly,
$stok_1$ and $stok_2$ are local variables of $\Proc{2}$ for $token_1$ and $token_2$. $(a,b) < (a',b')$ denotes lexicographic less than relation, i.e. 
$(ftok_2,2)<(ftok_1,1)$ \emph{iff} $ftok_2<ftok_1$ and $(stok_1,1)<(stok_2,2)$ \emph{iff} $stok_1\le stok_2$. 
This algorithm with inline assertions is shown in Figure \ref{fig:ex3} where $Inv1$ and $Inv2$ are as in Figure \ref{fig:ex3inv}. Assertions 
are on the history $\mathrm{h}_i$ and local variables of that process. History $\mathrm{h}_i$ is abstracted using the predicates of Figure \ref{fig:predhist}.
It should be noted that $V_1,W_1$ and $X_1$ do not appear explicitly inside the regular structure of the history abstracted by $[Q_1;R_1;T_1;U_1;\None!;P_1]^*$.
They appear in the last iteration of the loop and therefore the variables updated inside the loop, $(ftok_1,ftok_2)$, are assigned 
the placeholders corresponding to these reads, i.e. $\ph_2$, $\ph_3$ and $\ph_4$. Same holds for the invariant of $\Proc{2}$ as well. Subsequently, these elements are 
abstracted to $\None!$ in order to establish the loop invariant.

\begin{description}
\item [Mutual Exclusion proof]
 In order to prove the mutual exclusion property, we first state the required assertion to capture this property. 
\[
 \begin{array}{l}
ME \defeq \left(Inv1''(h_1) \land Inv2''(h_2) \land \compat(0,0,\false,\false,h_1,h_2) = \false\right)
 \end{array}
\]
where $Inv1''(h_1)$ and $Inv2''(h_2)$ are the assertions inside critical sections of $\Proc{1}$ and $\Proc{2}$ respectively. 
We prove it in two steps. First we show that,
\begin{equation*}\label{proof:thrd:1}
\begin{array}{l}
Inv1'(h_1) \land Inv2'(h_2) \land~ \compat(0,0,\false,\false,h_1,h_2) \implies Inter.~
\\~~Inter \defeq \inpar {Inter1 \defeq (ftok_2=0 \implies stok_1\neq0 \land stok_1< stok_2) 
\\~~~~~\land Inter2 \defeq (stok_1=0 \implies ftok_2\neq0 \land ftok_2<ftok_1)
\\~~~~~\land Inter3 \defeq (ftok_2\neq0 \land stok_1\neq0 \implies 
\\~~~~~~~~~~~~~~~~~~~~~~~~~~ftok_1 \le ftok_2 \implies stok_1\le stok_2)}
\end{array}
\end{equation*}

and then it is easy to show that,
\begin{equation*}\label{proof:thrd:2}
\begin{array}{l}
Inter \land \neg(ftok_2\neq0 \land (ftok_2,2)<(ftok_1,1)) 
\\~~~~~~~\land \neg(stok_1\neq0 \land (stok_1,1)<(stok_2,2))  = \false
\end{array}
\end{equation*}

\item[Proof of $Inv1'(h_1) \land Inv2'(h_2) \land \compat(0,0,\false,\false,h_1,h_2)\implies Inter1$],
We show that in two steps,
\begin{equation}\label{prf:1}
 \begin{array}{l}
 Inv1'(h_1) \land Inv2'(h_2) \land \compat(0,0,\false,\false,h_1,h_2)
\\~~~\land ftok_2=0 \implies stok_1\neq0 
 \end{array}
\end{equation}
and 
\begin{equation}\label{prf:2}
 \begin{array}{l}
 Inv1'(h_1) \land Inv2'(h_2) \land \compat(0,0,\false,\false,h_1,h_2) 
\\~~~\land ftok_2=0 \implies stok_1< stok_2 
 \end{array}
\end{equation}

\item[Proof of \eqref{prf:1}],
First assume that $Inv1'(h_1)$, $Inv2'(h_2)$, $\compat(0,0,\false,\false,h_1,h_2)$ and $ftok_2=0$ hold. 
Only way to have $ftok_2=ph_4=0$ is to put last read of $token_2$ in $h_1$ (denoted by $X_1$) in the merged history 
after $P_2$ of, say $k-1^{th}$ iteration of $[Q_2;R_2;T_2;U_2;\None!;P_2]$ and before $T_2$ of $k^{th}$ iteration. 
$Inv1$ implies that \ul{$T_1\before X_1$} hence the value of $token_1$ visible at $X_1$ is non-zero which also 
becomes visible at $T_2$ because of the placement of $X_1$ before $T_2$ and the fact that the history $\h{\Proc{1}}$ does not 
 contain any write to $token_1$ after $X_1$.
Further because $Inv2$ implies \ul{$T_2\before W_2$} therefore the read of $token_1$ at $W_2$ in $\Proc{2}$
also observes this non-zero value of $token_1$ and assigns it to $stok_1$. Therefore we have,
\begin{equation*}
 \begin{array}{l}
 Inv1(h_1) \land Inv2(h_2) \land \compat(0,0,\false,\false,h_1,h_2)
\\~~~\land ftok_2=0 \implies stok_1\neq0 
 \end{array}
\end{equation*}

Out of all the orderings of $\Proc{1}$ and $\Proc{2}$ only $T_1\before X_1$ and $T_2\before W_2$ 
are used to prove this part of the proof. Hence these two are \textit{sufficient} to prove \eqref{prf:1}.


\item[Proof of \eqref{prf:2}]
In order to have $ftok_2=ph_4=0$,  $X_1$ should be placed in the merged history 
after $P_2$, 
of say $k-1^{th}$ iteration of $[Q_2;R_2;T_2;U_2;\None!;P_2]$ and before $T_2$ in $k^{th}$ iteration.
Further, $Inv_1'$ implies $ph_2=\false$ and therefore $V_1$(corresponding to the read of $taking_2$ in last iteration) 
must be placed before $Q_2$ of any $n^{th}$ iteration and 
after $U_2$ of $n-1^{th}$ iteration such that $n\le k^{th}$. It must be noticed that $V_1$ cannot be put
after $U_2$ of $k^{th}$ iteration because \ul{$V_1\before X_1$} in the history of $\Proc{1}$.
$Inv1$ implies \ul{$T_1\prec V_1$} and hence the value of $token_1$ visible at 
$V_1$ is non-zero, say $v$. As $V_1$ has been placed before $Q_2$ of $n^{th}$ iteration, where $n\le k$,
hence the same value $v$ of $token_1$ is also visible at $Q_2$ of $n^{th}$ iteration.
$Inv2$ implies that \ul{$Q_2\before R_2$}
hence the same value $v$ of $token_1$ is also visible at $R_2$ of any subsequent iterations. $Inv2$ implies \ul{$R_2\before T_2$} 
and therefore all subsequent iterations of $T_2$ write $v+1$ to $token_2$ resulting in $stok_2=v+1$. 
Further the value visible at $stok_1$ in $\Proc{2}$ is still the last value written to $token_1$ by $\Proc{1}$, i.e. $v$. Therefore we get,
\begin{equation*}
 \begin{array}{l}
 Inv1'(h_1) \land Inv2'(h_2) \land \compat(0,0,\false,\false,h_1,h_2)
\\~~~\land ftok_2=0 \implies stok_1<stok_2
 \end{array}
\end{equation*}

Only orderings used in this part of the proof are $V_1\before X_1$, $T_1 \before V_1$, $R_2\before T_2$ and $Q_2 \before R_2$. Hence, out of 
all total orders in $h_1$ and $h_2$ these are \textit{sufficient} to prove \eqref{prf:2}.
%
%

\item[Proof of $Inv1' \land Inv2' \land \compat(0,0,\false,\false,h_1,h_2)\implies Inter2$,]
This is symmetric to previous proof and gives us following symmetric \textit{sufficient orderings}; $T_2\before X_2$, 
$T_1\before W_1$, $V_2\before X_2$, $T_2\before V_2$, $Q_1\before R_1$ and $R_1\before T_1$.

\item [Proof of $Inv1' \land Inv2' \land \compat(0,0,\false,\false,h_1,h_2)\implies Inter3$],
For $ftok_2\neq 0$, $X_1$ must read the non-zero value written to $token_2$ at $T_2$ and similarly for $stok_1\neq 0$,
$W_2$ must read the non-zero value written to $token_1$ at $T_1$. Let $X_1$ read from $T_2$ of 
iteration $k_2$ and $W_2$ reads from $T_1$ of iteration $k_1$. Following possibilities arise based on 
whether or not $k_1$ and $k_2$ are last iterations of $\Proc{1}$ and $\Proc{2}$.
\begin{description}
 \item [$k_1$ is not the last iteration and $k_2$ is the last iteration ]
In order to have $stok_1\neq0$, $W_2$ is placed after $T_1$ of $k_1^{th}$ iteration and before $P_1$ of $k_1^{th}$ iteration in the merged history. 
$Inv2$ implies that \ul{$T_2\before W_2$} and therefore the value written to $token_2$  in the last iteration of $\Proc{2}$
at $T_2$ is some $v$ such that $v\neq 0$ and it flows to $W_2$. $Inv1$ implies \ul{$P_1\before X_1$} and because $\Proc{1}$ does not write to 
$token_2$ hence $X_1$ also sees the value of $token_2$ as $v\neq 0$. 
$Inv1$ further implies that 
\ul{$P_1\before R_1$} for $R_1$ of any iteration greater than $k_1$ and \ul{$R_1\before T_1$} such that 
$R_1$ reads the value of $token_2$ and writes back this value incremented by 1 
to $token_1$ at $T_1$. Therefore $v+1$ is written to $token_1$ at $T_1$ in all iterations greater than $k_1$. 
Further, \ul{$T_1\before W_1$} implies that the same value $v+1$ is also visible at $W_1$.
Therefore we get $ftok_1>ftok_2$  implying $Inter3$ trivially holds.

\item [$k_1$ is the last iteration and $k_2$ is not the last iteration]
In order to have $ftok_2\neq 0$, $X_1$ is placed after $T_2$ of $k_2$ and before $P_2$ of $k_2$ in the merged history. 
$Inv1$ implies \ul{$T_1\before X_1$} and
therefore the value written to $token_1$ in the last iteration of $\Proc{1}$ at $T_1$ is some $v$ such that $v\neq 0$ and it flows to $X_1$.
$Inv2$ implies \ul{$P_2\before R_2$} for $R_2$ of any iteration greater than $k_2$ and \ul{$R_2\before T_2$}  such that $R_2$ reads the 
value of $token_1$ and writes back this value incremented by 1 to $token_2$ at $T_2$. Therefore $v+1$ is written to $token_2$ at 
$T_2$ in all iteration greater than $k_2$. Also $Inv2$ implies \ul{$T_2\before X_2$} which gives $stok_2$ equal to $v+1$. 
Further $Inv2$ also implies \ul{$P_2\before W_2$} 
hence the same value $v$ of $token_1$ which is visible at $P_2$ is assigned to $stok_1$. This results in $stok_1<stok_2$ and hence proved.

\item [Both $k_1$ and $k_2$ are last iterations]: In this case $X_1$ is placed after $T_2$ of the last iteration and because \ul{$T_2\before X_2$} hence 
$ftok_2=stok_2$. Similarly, if $W_2$ is placed after $T_1$ of the last iteration then from \ul{$T_1\before W_1$} in $Inv1$ we have $stok_1=stok_2$. 
Therefore $ftok_1\le ftok_2 \implies stok_1\le stok_2$ hence proved.

\item [Neither of $k_1$ and $k_2$ are last iterations]: We show that if $ftok_2$ reads the value of $token_2$ from $k_2^{th}$ iteration of $\Proc{2}$
 then it is not possible for $stok_1$ to read the value of $token_1$ from the iteration $k_1$ of $\Proc{1}$ which is not the last iteration. 
 Let $X_1$ is placed after $T_2$ and before $P_2$ of $k_2^{th}$ iteration. \ul{$T_1\before X_1$} implies 
that the value of $token_1$ visible at $X_1$ is from the last iteration of $T_1$ which also becomes visible at $P_2$ of $k_2^{th}$ iteration because of the placement of $X_1$. Further, \ul{$P_2\before W_2$} implies that $W_2$ sees the same value of $token_1$ written by the last iteration of $\Proc{1}$ at $T_1$ which is not $k_1$. 
Similar argument follows for $W_2$ as well. Hence no compatible merged history exists for this case.
\end{description}

Finally we collect the \textit{sufficient} orderings used to prove $Inter3$. 
\begin{itemize}
\item $T_1\before X_1$, $T_1\before W_1$, $P_1\before X_1$, $P_1\before R_1$ and $R_1\before T_1$ for $\Proc{1}$.
\item $T_2\before X_2$, $T_2\before W_2$, $P_2\before W_2$, $P_2\before R_2$ and $R_2\before T_2$ for $\Proc{2}$.
\end{itemize}
In $P_1\before R_1$ and $P_2\before R_2$, $R_i$ comes from the iteration later than that of $P_i$.

\paragraph{\textbf{Lamport's bakery algorithm under PSO memory model}}
PSO memory model allows write instructions in a process to be reordered with future write and read instructions 
operating on different addresses. We have following \textit{sufficient} instruction orderings needed to prove $Inter$ 
and hence mutual exclusion.
\begin{itemize}
 \item $T_1\before V_1$, $T_1\before W_1$, $T_1\before X_1$, $P_1\before X_1$, $V_1 \before X_1$, $Q_1\before R_1$, $R_1\before T_1$
 \item $T_2\before V_2$, $T_2\before X_2$, $T_2\before W_2$, $P_2\before W_2$, $V_2 \before X_2$, $Q_2\before R_2$, $R_2\before T_2$
 \item $P_1\before R_1$ if $P_1$ is from iteration $k$ then $R_1$ is from iteration greater than $k$.
 \item $P_2\before R_2$ if $P_2$ is from iteration $k$ then $R_2$ is from iteration greater than $k$.
 \end{itemize}
The PSO memory model preserves the ordering of any two instructions which are data or control dependent. Therefore $T_1\before W_1$, $T_2\before X_2$,
$R_1\before T_1$ and $R_2\before T_2$ are also satisfied by PSO. Further, PSO also preserves the order of two read instructions, i.e. 
$V_1\before X_1$ and $V_2\before X_2$. A fence between $T_1$ and $V_1$ also satisfies the ordering $T_1\before X_1$ and 
symmetrically a fence between $T_2$ and $V_2$ also satisfies the ordering $T_2\before W_2$.
Further, a fence between $Q_1$ and $R_1$ also orders $P_1$ of $k^{th}$ iteration before $R_1$ of $\ge k+1^{th}$ iterations and $X_1$.
A fence between $Q_2$ and $R_2$ also orders $P_2$ of $k^{th}$ iteration before $R_2$ of $\ge k+1^{th}$ iteration and $W_2$.
Therefore we only need two fence instructions per process, i.e. between $Q_1$ and $R_1$ and between $T_1$ and $V_1$ in $\Proc{1}$ and symmetrically in $\Proc{2}$.
\end{description}

%
%
%
%

\newcommand{\reading}{\mathit{reading}}
\newcommand{\latest}{\mathit{latest}}
\renewcommand{\index}{\mathit{index}}
\newcommand{\slot}{\var{slot}}
\newcommand{\profiveone}{
\assh{Inv1\defeq \lambda h.~}\\
\assh{\inpar{
\Let ph_1,ph_2,P_1,Q_1,R_1,S_1,T_1 \In~\\
P_1 = \lambda h.~ \r{reading}{ph_1}(h) \In\\
Q_1 = \lambda h.~ \r{index[\neg ph_1]}{ph_2}(h) \In\\
R_1 = \lambda h.~ \w{slot[\neg ph_1][\neg ph_2]}{\_}(h) \In\\
S_1 = \lambda h.~ \w{index[\neg ph_1]}{\neg ph_2}(h) \In\\
T_1 = \lambda h.~ \w{latest}{\neg ph_1}(h) \In\\
~~~~[P_1;Q_1;R_1;S_1;T_1]^*(h) \land \None!\reading(h)
}}
}

\newcommand{\profivetwo}{
\assh{Inv2\defeq \lambda h.~}\\
\assh{\inpar{
\Let ph_1',ph_2',P_2,Q_2,R_2,S_2. \In~\\
P_2 = \lambda h.~ (?latest,ph_1')(h) \In\\
Q_2 = \lambda h.~ (!reading,ph_1')(h) \In\\
R_2 = \lambda h.~ (?index[ph_1'],ph_2')(h) \In\\
S_2 = \lambda h.~ (?slot[ph_1'][ph_2'],\_)(h) \In\\
~~~~[P_2;Q_2;R_2;S_2]^*(h) \land \None!\index(h)\\	
~~~~~~~\land \None!\latest(h) \land \None!\slot(h)
}}
}

\newcommand{\MergedHist}{\mathrm{MergedHist}}
\newcommand{\Reader}{\mathrm{Reader}}
\newcommand{\Writer}{\mathrm{Writer}}

\newcommand{\n}[1]{\overline{#1}}
\newcommand{\call}[3]{#3_{#1^{#2}}}

\newcommand{\FigureExV}{
\begin{figure}[h]
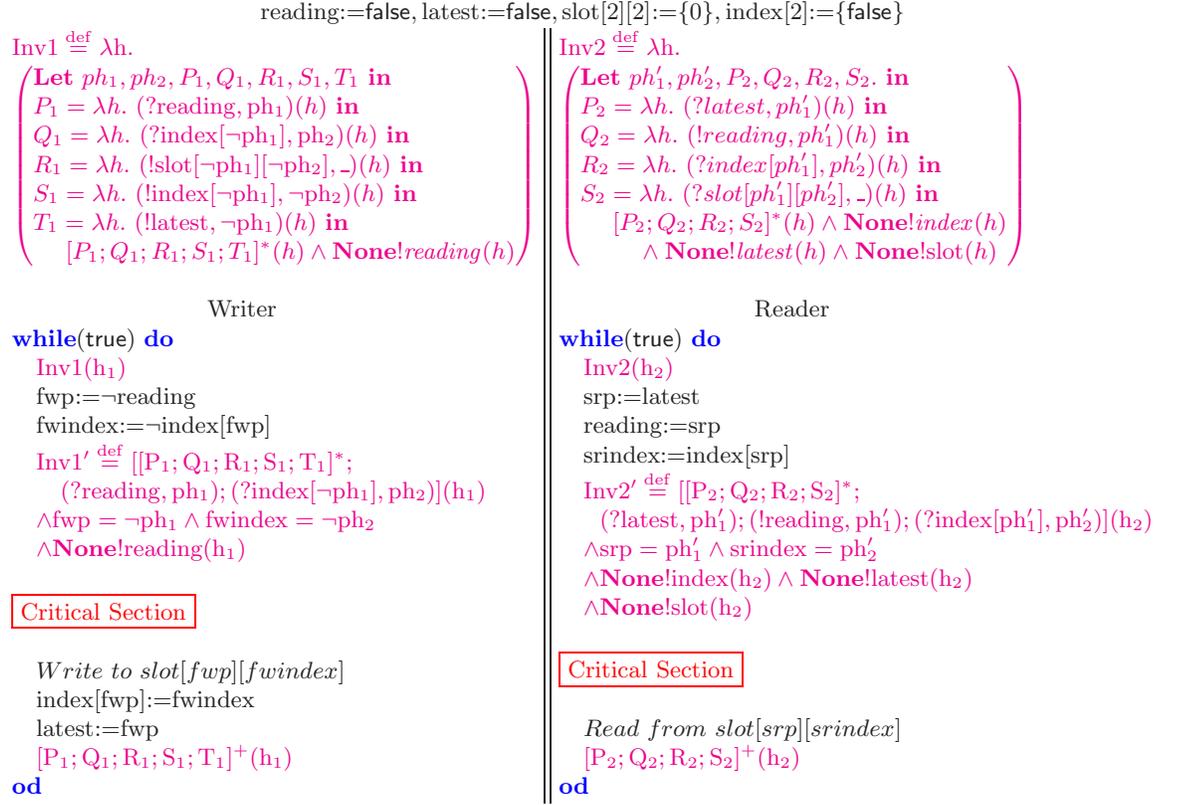

\centering
$\begin{array}{@{~}c@{~}}
  \Assign{reading}{\false}, \Assign{latest}{\false}, \Assign{slot[2][2]}{\set{0}}, \Assign{index[2]}{\set{\false}}
\\
\begin{array}{@{}l@{~}||@{~}l@{}}
\inarr{
\profiveone\\ \\
\qquad \qquad \qquad \qquad \Writer\\
\CWhile{\true}\\
\quad \assh{Inv1(h_1)}\\
\quad \Assign{fwp}{\neg reading}\\
\quad \Assign{fwindex}{\neg index[fwp]}\\ 
\quad \assh{Inv1'\defeq[[P_1;Q_1;R_1;S_1;T_1]^*;}\\
\quad \assh{~~~(?reading,ph_1);(?index[\neg ph_1],ph_2)](h_1)}\\
\quad \assh{\land fwp=\neg ph_1 \land fwindex=\neg ph_2}\\
\quad \assh{\land \None!reading(h_1)}\\ \\
\CS \\ \\
\quad Write\ to\ slot[fwp][fwindex]\\
\quad \Assign{index[fwp]}{fwindex}\\
\quad \Assign{latest}{fwp}\\
\quad \assh{[P_1;Q_1;R_1;S_1;T_1]^+(h_1)}\\
\COd
}
& 
\inarr{
\profivetwo\\ \\
\qquad \qquad \qquad \qquad \Reader\\
\CWhile{\true}\\
\quad \assh{Inv2(h_2)}\\
\quad \Assign{srp}{latest}\\
\quad \Assign{reading}{srp}\\
\quad \Assign{srindex}{index[srp]}\\
\quad \assh{Inv2'\defeq[[P_2;Q_2;R_2;S_2]^*;}\\
\quad \assh{~~(?latest,ph_1');(!reading,ph_1');(?index[ph_1'],ph_2')](h_2)}\\
\quad \assh{\land srp=ph_1' \land srindex=ph_2'}\\
\quad \assh{\land \None!index(h_2) \land \None!latest(h_2)}\\
\quad \assh{\land \None!slot(h_2)}\\ \\
\CS\\ \\
\quad Read\ from\ slot[srp][srindex]\\
\quad \assh{[P_2;Q_2;R_2;S_2]^+(h_2)}\\
\COd
}
\end{array}
\end{array}$
 \caption{Simpson's 4 slot algorithm }
 \label{fig:ex5}
\end{figure}
}

\subsection{Example: Simpson's 4 slot algorithm \cite{Simpson4slot}}
This is a wait-free algorithm for concurrent access of a location from a single reader and a single writer processes. This location is 
simulated using a 2$\times$2 array variable $\slot$. If the
reader is reading the data and the writer wants to write new data at the same time then 
instead of waiting for the reader to complete, the writer writes to a different index of $\slot$ and 
indicates the reader to read from this location in the subsequent read. 
Boolean variables $reading$ and $latest$ represent two indices ($\false$ as 0 and $\true$ as 1) to denote the row and column 
indices of $\slot$ and $\index$ variable. Variable $\index$ is a two element boolean array such that for any $s\in \set{\true,\false}$, 
$\slot[s][index[s]]$ has the latest data written by the writer in the row $\slot[s]$. Variables $fwp$ and $fwindex$ are local to the writer process. 
Variables $srp$ and $srindex$ are local to the reader process. The algorithm with inline assertions 
on history and local variables are shown in Figure \ref{fig:ex5}. In order to simulate different invocations of the writer and the reader 
the respective program is enclosed within the loop. We are interested in proving following two properties of this algorithm.

\FigureExV

\begin{description}
 \item [Interference Freedom]
We want to show that at the entry point of critical sections (for the writer and the reader) $fwp\neq srp$, 
i.e. the reader and the writer use different rows of the $\slot$ variable, and 
$fwp=srp \limp fwindex\neq srindex$, i.e. if both use the same row of the $\slot$ variable then 
they read from and write to different column of that row. 
Therefore, we want to prove 
\[
\begin{array}{l}
 Inv1'(h_1) \land Inv2'(h_2) \land \compat(\false,\false,\set{0},\set{\false},h_1,h_2) \implies \\
~~~fwp\neq srp \lor fwp=srp \limp fwindex\neq srindex
\end{array}
\]
Where $Inv1'$ and $Inv2'$ are the assertions just before the program point where writer is going to write the data and reader is going to read the data.
\begin{proof}
In any compatible merged history $h$ of $h_1$ and $h_2$ the placement of the last write $\w{reading}{ph_1'}$ of $h_2$ has 
two choices with respect to the last read $\r{reading}{ph_1}$ of $h_1$.
\begin{itemize}
 \item \textit{Last write $\w{reading}{ph_1'}$ of $h_2$ is placed \ul{before} the last read $\r{reading}{ph_1}$ of $h_1$:} As writer does not 
write to $reading$ hence $ph_1=ph_1'$, or $\neg fwp=srp$ (from assertions $fwp=\neg ph_1$ and $srp=ph_1'$) and therefore $fwp\neq srp$. Hence proved.

 \item \textit{Last write $\w{reading}{ph_1'}$ of $h_2$ is placed \ul{after} the last read $\r{reading}{ph_1}$ of $h_1$:} This, along with the assumption $fwp=srp$ implies
$\neg ph_1 = ph_1'$. Therefore in $Inv2'$, $\r{index[ph_1']}{ph_2'}$ is same as $\r{index[\neg ph_1]}{ph_2'}$. According to $Inv1'$ 
$\r{index[\neg ph_1]}{ph_2}$ is in $h_1$. We want to establish the relation between $ph_2$ and $ph_2'$.

$Inv1'$ implies that \ul{$S_1\before \r{reading}{ph_1}$} hence there is no write to variable $\index$ beyond $\r{reading}{ph_1}$ in the merged history. 
Hence the value at $index[\neg ph_1]$ is same for both processes, i.e. $ph_2=ph_2'$. From $Inv1'$ and $Inv2'$ we get $fwindex=\neg ph_2$ and $srindex=ph_2$
which implies $fwindex\neq srindex$. Hence proved. 
\end{itemize}
\end{proof}
From the above proof, we can see that the only ordering important in proving this property is $S_1\before P_1$ where $P_1$ is from later iteration than that of $S_1$. Now
we prove another property of interest and find out the program orderings used in that proof.

\item [Consistent Reads]
Second property of interest is related to the order of reads. It specifies that the values read by the reader form a stuttering sequence of values written by the writer, i.e. if the writer writes the sequence 1,2,3,4,5 in subsequent invocations of write then the reader cannot observe 1,4,3,5 as a sequence read. It must observe the sequence which preserve the order of writes and possibly interspersed with the repetition of the same data. First, we define some notations used in this section.
Let $\MergedHist(\Reader^R,\Writer^W)$ be the set of all compatible merged histories consisting of $R$ invocations of the reader process and $W$ invocations of the writer process. 
Let $\Reader(n)$ and $\Writer(n)$ be the $n^{th}$ invocations of the reader and the writer respectively. Let $\Data{k}$ be the data written by the writer in the $k^{th}$ invocation. 
Let $D(w)$ be a sequence $\Data{1}.\Data{2}.\cdots.\Data{w}$ of values written by $w$ consecutive iterations of the writer. For $s\in \set{\true,\false}$, $\n{s}$ denote the negation of $s$.
Let $\call{Reader}{r}{elem}$ be the element $elem\in \set{\r{\any}{\any},\w{\any}{\any}}$ in the merged history from $r^{th}$ invocation of the $\Reader$. Similarly $\call{Writer}{w}{elem}$ denote the same for for the $w^{th}$ iteration of the $\Writer$.

\paragraph{\textbf{Stuttering sequence}}
 Let $\Seq{r,w}$ be a sequence of length $r$, constructed from the elements of $D(w)$. $\Seq{r,w}$ is a stuttering sequence of $D(w)$ if for any index $i$ of the sequence $\Seq{r,w}$ such the $\Seq{r,w}[i-1]=\Data{k_1}$ and $\Seq{r,w}[i]=\Data{k_2}$ then $k_1 \le k_2$. 
 
%

\paragraph{\textbf{Some interesting properties of the Reader and the Writer processes}}
Only the writer writes to $latest$ and reads from $reading$ variable. Further, only the reader writes to $reading$ and reads from $latest$. Also, the value written to $reading$ by the reader in any invocation is same as the value read from $latest$. The value written to $latest$ by the writer in any invocation is negation of the value that it reads from $reading$. 

Following lemma characterizes the sequence of values written to $\reading$ in a segment of the merged history. This characterization is then used in the proof of Lemma \ref{finalaim}. 
\begin{lemma}\label{one}
For all $R$, $W$, $h\in \MergedHist(\Reader^R,\Writer^W)$, $r\le R, w\le W$, $s\in \set{\true,\false}$, if $\Reader(r)$ reads the value of $\latest$ written by $\Writer(w)$ as $s$ then in the sequence of values written to variable $\reading$ between $P_1$ of $\Writer(w)$ and $P_2$ of $\Reader(r)$, $s$ is never followed by $\n{s}$.
\end{lemma}
 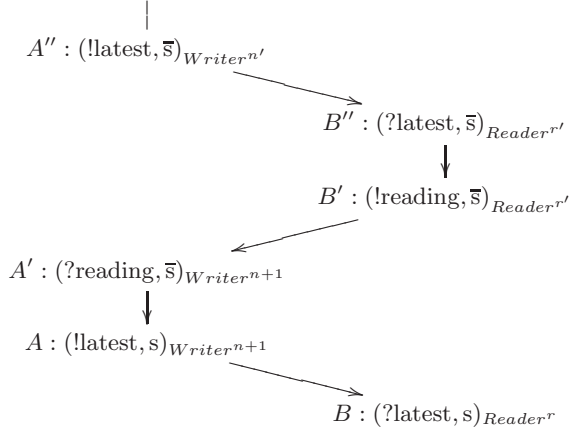
\begin{figure}[h]
\begin{tabular}{@{}c@{~}}
$\xymatrix @R=1pc @C=0.5pc {
\ar@{--}[d]\\
A'':\call{Writer}{n'}{\w{latest}{\n{s}}} \ar[dr]_{}             &  \\
 & B'':\call{Reader}{r'}{\r{latest}{\n{s}}} \ar[d]_{} \\
 & B': \call{Reader}{r'}{\w{reading}{\n{s}}} \ar[dl]_{}\\
A':\call{Writer}{n+1}{\r{reading}{\n{s}}} \ar[d]_{}\\
A:\call{Writer}{n+1}{\w{latest}{s}} \ar[dr]_{}\\
 &B:\call{Reader}{r}{\r{latest}{s}}\\
}$
\end{tabular}
\caption{Merged history for Lemma \ref{one} }
\label{lem:fig1}
\end{figure}

\begin{proof}
We prove it using induction on the iteration number of the writer process. \\
\textbf{Base case, $w=0$:} If $\Reader(r)$ reads the initial value of $\latest$ (in 0$^{th}$ iteration of the writer process), say $s$, then all the invocations of the reader before $\Reader(r)$ also see this initial value of $\latest$ as $s$ and therefore the sequence is made of only this value. Hence the base case satisfies this property.\\
\textbf{Induction Hypothesis, $w\le n$:} For all $w\le n$, if $\Reader(r)$ reads the value of $\latest$ as $s$, written by $\Writer(w)$ then the sequence of values written to $\reading$ between 
$P_1$ of $\Writer(w)$ and $P_2$ of $\Reader(r)$ satisfies this property.\\
\textbf{Induction Step, $w=n+1$:} Consider the merged history of Figure \ref{lem:fig1} where $\Reader(r)$ reads the value of $\latest$(denoted by $B$) from $\Writer(n+1)$ (denoted by $A$). We want to prove this property for the sequence of values written to $\reading$ between $A'$ and $B$ where $A'$ denotes $P_1$ of $\Writer(n+1)$ in Figure \ref{lem:fig1}. \ul{$P_1\before T_1$} implies that $A'$ appears before $A$ in the merged history. Let $\Writer(n+1)$ read the value of $\reading$ from $\Reader(r')$ (denoted by $B'$). \ul{$P_2\before Q_2$} implies that $P_2$ of $\Reader(r')$ appears before $B'$ in the merged history. Let $\Reader(r')$ at $P_2$ read the value of $\latest$ from $\Writer(n')$($A''$). It is clear that $n'<n+1$ hence from the induction hypothesis we know that in the sequence of values written to $\reading$ between $A''$ and $B''$, $\n{s}$ is never followed by $s$. We use this knowledge to characterize the values written to $\latest$ between $B''$ and $A'$. From the writer process we know that the value written to $\latest$ is negation of the value read from $\reading$ variable. Therefore in the sequence of values written to $\latest$ between $B''$ and $A'$, $s$ is never followed by $\n{s}$. Further, we also know that the reader process writes the same value to $\reading$ as is read from the  variable $\latest$. Therefore in the sequence of values written to $\reading$ between $A'$ and $B$, $s$ is never followed by $\n{s}$. Hence proved.
\end{proof}

For any reader we establish the relation between its read of variable $\latest$ and the read of data. More formally we want to say that,
\begin{lemma}\label{finalaim}
 For all $R$, $W$, $h\in \MergedHist(\Reader^R,\Writer^W)$, $r\le R, w\le W$, $s\in \set{\true,\false}$, if $\Reader(r)$ reads the value of $latest$ as $s$ written by $\Writer(w)$ 
 and reads the value of $\index[s]$ written by $\Writer(w')$ then $w'\ge w$ and $\Reader(r)$ reads the data $\Data{w'}$ and 
 all the invocations of the writer from $w$ to $w'$ write the data in the row $s$ of the $\slot$ variable.
\end{lemma}
\begin{proof}
Let us assume that $\Reader(r)$ reads the value of $\latest$ as $s$ written by $\Writer(w)$ which means that $P_2$ of $\Reader(r)$ appears after $T_1$ of $\Writer(w)$ and no write to $\latest$ appears in between these two points in the merged history. As $\Writer(w)$ also writes to $\index[s]$ and \ul{$S_1\before T_1$} and \ul{$P_2\before R_2$}, hence $\Reader(r)$ reads the value of $\index[s]$ either from $\Writer(w)$ or from $\Writer(w')$ such that $w'\ge w$. We prove the following two properties,
\begin{itemize}
 \item \textbf{$\mathbf{\Reader(r)}$ reads the data written by $\Writer(w')$ in $\slot[s][k]$ where $k$ is the value written to $\index[s]$ by $\Writer(w')$} \\
\emph{Proof:} If $\Reader(r)$ reads $\index[s]$ as $k\in \set{\true,\false}$ from $\Writer(w')$ then \ul{$R_1\before S_1$} implies that $\Writer(w')$ has also written the data in $\slot[s][k]$. We want to show that 
no subsequent invocation of the writer writes in $\slot[s][k]$ before $S_2$ of $\Reader(r)$. From the assumption that $R_2$ of $\Reader(r)$ reads the value of $\index[s]$ from $\Writer(w')$ implies that there is no write to $\index[s]$ between $S_1$ of $\Writer(w')$ and $R_2$ of $\Reader(r)$. \ul{$S_1 \before P_1$}, where $P_1$ is from later iterations than 
that of $S_1$, in $Inv1'$ implies that the writer can be invoked at most once after the iteration $w'$ and before $R_2$ of $\Reader(r)$. If invoked more than once it results in the write of $\index[s]$ to appear after $S_1$ of $\Writer(w')$ and before $R_2$ of $\Reader(r)$, contradicting our assumption. Further, if the next invocation after $w'$ happens then it writes the data to column $\n{k}$ of row $s$ in $\slot$ variable still satisfying the property. If any further invocation of the writer happens after $R_2$ and before $S_2$ of $\Reader(r)$ then because of \ul{$Q_2\before R_2$} it observes the 
value of $\reading$ as $s$ and therefore writes the data to $\n{s}$ row of the $\slot$ variable. $Inv2'$ implies that \ul{$S_2$ of $\Reader(r) \before Q_2$} of subsequent invocations of the reader and therefore no write to $\reading$ exists between $R_2$ and $S_2$ of $\Reader(r)$. This results in observing the same value of $\reading$, $s$, and subsequently writing the data to row $\n{s}$ by all invocations of the writer between $R_2$ and $S_2$ of $\Reader(r)$.

\item \textbf{All the invocations of the writer from $w$ to $w'$ write the data in the row $s$ of the $\slot$ variable}\\
\emph{Proof:} We prove an alternate equivalent property; All invocations of the writer from $w$ to $w'$ read the value of $\reading$ as $\n{s}$. This follows from the property that the 
writer writes in that row of the $\slot$ variable that is obtained by negating the value of the variable $\reading$. From the assumption, $\Writer(w')$ writes to $\index[s]$ and therefore it also reads the value of $\reading$ as $\n{s}$ and the same holds for $\Writer(w)$ as well. We need to show that this property holds for the rest of the invocations in between $w$ and $w'$. We prove it by contradiction by assuming that there exists a $w''$ such that $w<w''<w'$ and $P_1$ of $\Writer(w'')$ reads the value of $\reading$ as $s$. It is clear that $P_1$ of $\Writer(w'')$ appears after $P_1$ of $\Writer(w)$ and before $P_1$ of $\Writer(w')$. From the earlier argument we know that the value of $\reading$ visible at $P_1$ of $\Writer(w)$ is $\n{s}$. Hence, $P_1$ of $\Writer(w'')$ cannot read from the value of $\reading$ visible at $P_1$ of $\Writer(w)$. Therefore, it must read the value of $\reading$ as $s$ from the sequence of values written to $\reading$ between $P_1$ of $\Writer(w)$ and $P_2$ of $\Reader(r)$. 
Following Lemma \ref{one}, $s$ can never be followed by $\n{s}$ in the sequence of values written to $\reading$ in between $P_1$ of $\Writer(w)$ and $P_2$ of $\Reader(r)$. 
It further implies that $P_1$ of $\Writer(w')$ cannot read the value of $\reading$ as $\n{s}$, a contradiction to our assumption. Therefore all invocations of the writer from $w$ to $w'$ read the value of $\reading$ as $\n{s}$. Hence proved.
\end{itemize}

\end{proof}

Now we use Lemma \ref{finalaim} to prove the following theorem.
\begin{theorem}[Stuttering Sequence]
Stutter(n)$\defeq$ For all $n,w \in \mathbb{N},h\in \MergedHist(\Reader^n,\Writer^w)$ the sequence
$\Seq{n,w}$, constructed from the elements of $D(w)$, is a stuttering sequence of $D(w)$.
\end{theorem}
\begin{proof}
 \textbf{Base case, $n=0$}: With no history corresponding to the Reader in the merged history, the empty sequence trivially forms a stuttering sequence.\\
 \textbf{Induction Hypothesis}: For all $r \le n$ Stutter(r) holds true.\\
 \textbf{Induction Step, $r=n+1$}: 

Let $\Reader(n)$ read the value of $\latest$ from $\Writer(w_a)$.
From Lemma \ref{finalaim}, we know that $Reader(n)$
reads the value $\Data{w'}$ such that $w_a \le w' \le w$ and all the invocations of the writer from $w_a$ to $w'$ write the data in the same row $s$ of the $\slot$ variable.
Further, the value of $index[s]$ visible at $R_2$ of $\Reader(n)$ is from $S_1$ of $\Writer(w')$. We know that the write to $latest$ from consecutive iterations are totally ordered and 
the same holds for consecutive reads from $latest$ as well.  
Therefore, $\Reader(n+1)$ reads the value of $latest$ from some $\Writer(w_b)$ such that $w_b \ge w_a$. There are two possibilities based on whether this value is $s$ or $\n{s}$.
\begin{itemize}
\item $\Reader(n+1)$ reads the value of $latest$ as $s$ from the $\Writer(w_b)$ such that $w_b \ge w_a$: From our assumption we know that $\Reader(n)$ sees the value of $index[s]$ from 
$\Writer(w')$ therefore the value of $index[s]$ visible at $R_2$ of $\Reader(n+1)$ will be from some $w''$ such that $w\ge w''\ge w'$. 
Following Lemma \ref{finalaim},  the data read by 
$\Reader(n+1)$ from $slot[s][index[s]]$ will be from $\Writer(w'')$ and $w''\ge w'$ implies that the resulting sequence $\Seq{n+1,w}$ is still a stuttering sequence.

\item  $\Reader(n+1)$ reads the value of $latest$ as $\n{s}$ from $\Writer(w_b)$ such that $w_b \ge w_a$: From Lemma \ref{finalaim}, we know that all the invocations of the writer from 
$w_a$ to $w'$ write the data in the same row $s$ of $\slot$ and because $\Writer(w_b)$ writes the data in row $\n{s}$ of $\slot$ therefore $w_b > w'$. Following Lemma \ref{finalaim}, $Reader(n+1)$ reads the data from some $\Writer(w'')$ such that $w''\ge w_b$. Combining these two we get $w''>w'$ such that $\Reader(n)$ reads $\Data{w'}$ and $\Reader(n+1)$ reads $\Data{w''}$. Therefore the 
resulting sequence $\Seq{n+1,s}$ is still a stuttering sequence.
\end{itemize}

\end{proof}

\end{description}
\paragraph{\textbf{Simpson's 4 slot algorithm under PSO memory model}}
Following per-thread instruction orderings are used in the proofs of the interference freedom and the consistent reads properties.
\begin{itemize}
 \item $P_1\before T_1$(from the proof of Lemma \ref{one}), $S_1 \before T_1$, $R_1\before S_1$
 \item $P_2\before Q_2$(from the proof of Lemma \ref{one}), $P_2\before R_2$, $Q_2\before R_2$
 \item $S_1\before P_1$, where $P_1$ is from iteration later than that of $S_1$
 \item $S_2 \before Q_2$, where $Q_2$ is from iteration later than that of $S_2$.
\end{itemize}
Out of all these orderings, $P_1\before T_1$, $P_2\before Q_2$ and $P_2\before R_2$ are data dependent orders which are respected by the PSO memory model. 
$S_2\before Q_2$ where $Q_2$ is from iteration later than that of $S_2$ is also respected by PSO memory model because $S_2$ corresponds to read and $Q_2$ 
corresponds to write instruction. Therefore, we have to enforce only $R_1\before S_1$, $S_1\before T_1$, $Q_2\before R_2$ and $S_1\before P_1$ where $P_1$ is 
from iteration later than that of $S_1$. Following the semantics of $\fence$ it is sufficient to put two fence instructions in the writer; one between $R_1$ and $S_1$ 
and another between $S_1$ and $T_1$. Further, we need one fence instruction between $Q_2$ and $R_2$ in the reader as well. 

\section{Conclusion and Future Work}
\label{sec:conclusion}
In this paper we proved Simpson's 4 slot algorithm correct under the 
SC memory model with respect to the interference freedom and the consistent reads properties. Based on these proofs, we identified the locations of the
$\fence$ instructions needed to satisfy these two properties under the PSO memory model. 
As a direction for future work
we still have to explore the use of this approach for advanced memory models which support non-atomic writes (POWER/ARM). 
This paper introduces the predicates over history variable only to carry out proofs conveniently. No formal treatment is 
given to them and this should be addressed with high priority. Even in presence of predicates over history variables, the difficulty in carrying out these proofs 
without any tool support is evident from the proof of Simpson's 4 slot algorithm. 
We plan to use proof assistants for this in near future. In all the examples we 
considered here, a program is always executed by only one thread. This restriction 
needs to be addressed in order to handle concurrent data structures where many threads can execute the 
same method of an object.
\bibliographystyle{alpha}
\bibliography{paper}
\end{document}